\documentclass[12pt]{article}
\usepackage{amsmath, amsthm, amssymb, bigints,algorithm}

\usepackage[toc,page]{appendix}
\usepackage[pdftex]{graphicx}
\usepackage{pdfpages}
\usepackage{epstopdf}
\usepackage{authblk}
\usepackage{booktabs}
\usepackage{afterpage}
\usepackage{capt-of}
\usepackage{rotating}

\usepackage{bbm}
\usepackage{mathtools}
\usepackage{lipsum}
\usepackage{natbib}
\usepackage{xr}

\makeatletter
\newcommand*{\addFileDependency}[1]{
  \typeout{(#1)}
  \@addtofilelist{#1}
  \IfFileExists{#1}{}{\typeout{No file #1.}}
}
\makeatother

\newcommand*{\myexternaldocument}[1]{
    \externaldocument{#1}
    \addFileDependency{#1.tex}
    \addFileDependency{#1.aux}
}
\myexternaldocument{supp}

\newcommand{\blind}{1}

\newcommand{\be} {\begin{eqnarray*}}
\newcommand{\ee} {\end{eqnarray*}}
\newcommand{\ceil}[1]{\left\lceil #1 \right\rceil}

\newcommand{\Proj} {\mr{Q}}
\newcommand{\IP} {(\mr{I}-\mr{Q}_0)}

\newcommand{\bx} {{\bf x}}

\def\mr{\mathrm}

\newcommand{\Var}{\mathop{\rm Var}}

\newcommand{\norm}[1]{\big\Vert#1\big\Vert}

\newcommand\wt[1]{{ \widetilde{#1} }}

\newcommand \bbE{\mathbb{E}}

\def\T{{ \mathrm{\scriptscriptstyle T} }}

\def\Gauss{{ \mathrm{N} }}

\newcommand{\ind}{\mathbbm{1}}
\def\m{\mathcal}
\def\mb{\mathbb}

\renewcommand{\baselinestretch}{1.20}

\newtheorem{theorem}{Theorem}[section]
\newtheorem{lemma}[theorem]{Lemma}

\begin{document}
\def\spacingset#1{\renewcommand{\baselinestretch}%
{#1}\small\normalsize} \spacingset{1}


\if1\blind
{
  \title{\bf Functional Horseshoe Priors for Subspace Shrinkage}
  \author{
    Minsuk Shin
    
    Department of Statistics, Texas A\&M University\\
    
    Anirban Bhattachrya 
    
      Department of Statistics, Texas A\&M University\\
      
      and\\
    Valen E. Johnson 
    
      Department of Statistics, Texas A\&M University\\  
   }
   \date{}
  \maketitle
} \fi

\if0\blind
{
  \bigskip
  \bigskip
  \bigskip
  \begin{center}
    {\LARGE\bf Functional Horseshoe Priors for Subspace Shrinkage}
\end{center}
\date{}
  \medskip
} \fi

\bigskip
\begin{abstract}
 We introduce a new shrinkage prior on function spaces, called the functional horseshoe prior (fHS), that encourages shrinkage towards parametric classes of functions. Unlike other shrinkage priors for parametric models, the fHS shrinkage acts on the shape of the function rather than inducing sparsity on model parameters. We study the efficacy of the proposed approach by showing an adaptive posterior concentration property on the function. We also demonstrate consistency of  the model selection procedure that thresholds the shrinkage parameter of the functional horseshoe prior. We apply the fHS prior to nonparametric additive models and compare its performance with  procedures based on the standard horseshoe prior and several penalized likelihood approaches.  We find that the new procedure achieves smaller estimation error and more accurate model selection than other procedures in several simulated and real examples. The supplementary material for this article, which contains additional simulated and real data examples, MCMC diagnostics, and proofs of the theoretical results, is available online.
\end{abstract}

\noindent%
\vfill

\spacingset{1.45} 



%

\section{Introduction}
Since the seminal work of \citet{james1961estimation}, shrinkage estimation has been immensely successful in various statistical disciplines and continues to enjoy widespread attention. Many shrinkage estimators have a natural Bayesian flavor. For example, one obtains the ridge regression estimator as the posterior mean arising from an isotropic Gaussian prior on the vector of regression coefficients \citep{jeffries1961theory,hoerl1970ridge}. Along similar lines, an empirical Bayes interpretation of the positive part of the  James--Stein estimator can be obtained \citep{efron1973stein}. Such connections have been extended to the semiparametric regression context, with applications to smoothing splines and penalized splines \citep{wahba1990spline,ruppert2003semiparametric}. Over the past decade and a half, a number of second-generation shrinkage priors have appeared in the literature for application in high-dimensional sparse estimation problems. Such priors can be almost exclusively expressed as global-local scale mixtures of Gaussians \citep{polson2010shrink}; examples include 
the relevance vector machine \citep{tipping2001sparse}, normal/Jeffrey's prior \citep{bae2004gene}, the Bayesian lasso \citep{park2008bayesian,hans2009bayesian}, 
the horseshoe (HS) prior \citep{carvalho2010horseshoe}, normal/gamma and normal/inverse-Gaussian priors \citep{caron2008sparse,griffin2010inference}, generalized double Pareto priors 
\citep{armagan2013generalized} and Dirichlet--Laplace priors \citep{bhattacharya2015dirichlet}. These priors typically have a large spike near zero with heavy tails, thereby providing an approximation to the operating characteristics of sparsity inducing discrete mixture priors \citep{
george1997approaches,johnson2012bayesian}. For more on connections between Bayesian model averaging and shrinkage, refer to \cite{polson2010shrink}.

A key distinction between ridge-type shrinkage priors and the global-local priors is that while ridge-type priors typically shrink towards a fixed point--most commonly the origin--global-local priors shrink towards the union of subspaces consisting of sparse vectors. The degree of shrinkage to sparse models is  controlled by certain hyperparameters \citep{bhattacharya2015dirichlet}. In this article, we further enlarge the scope of shrinkage priors by proposing a class of functional shrinkage priors called  functional horseshoe (fHS) priors. fHS priors  facilitate shrinkage towards pre-specified subspaces. The shrinkage factor (defined in Section 3) is assigned a $\mbox{Beta}(a, b)$ prior with $a, b < 1$, which has the shape of a HS prior \citep{carvalho2010horseshoe}. While the HS prior  shrinks towards sparse vectors, the proposed fHS shrinks functions towards arbitrary subspaces. 

To illustrate the proposed methodology, consider a nonparametric regression model with unknown regression function $f : \m X \to \mb R$ given by
\begin{eqnarray}\label{eq:normal_mean}
Y = F + \varepsilon,  \quad \varepsilon \sim \Gauss(0, \sigma^2 \mr I_n),
\end{eqnarray}
where $Y = (y_1, \ldots, y_n)^{\T}$, $F = (f(x_1), \ldots, f(x_n))^{\T} = \bbE (Y \mid \bx)$,  and covariates $x_i \in \m X \subset \mb R$.  

In \eqref{eq:normal_mean}, one can either make parametric assumptions (e.g., linear or quadratic dependence on $x$) regarding the shape of $f$, or one may model it nonparametrically using splines, wavelets, Gaussian processes, etc. Scatter plots or goodness-of-fit tests can be used  to ascertain the validity of a linear or quadratic model in \eqref{eq:normal_mean}, but such procedures are only feasible in relatively simple settings. In relatively complex and/or high dimensional problems, there is clearly a need for an automatic data-driven procedure to adapt between models of varying complexity. With this motivation, we propose the fHS prior that encourages shrinkage towards a parametric class of models embedded inside a larger semiparametric model, as long as a suitable projection operator can be defined. The main difference between the fHS prior and the standard HS prior is that the fHS prior introduces a more general notion of shrinkage which operates on the shape of an unknown function rather than shrinking a vector of parameters to zero. We provide a more detailed discussion on this in Section 5.1.


The continuous nature of the prior allows development of a simple and efficient Gibbs sampler. As a consequence, the fHS procedure enjoys substantial computational advantages over traditional Bayesian model selection procedures based on mixtures of point mass priors, since they require computationally intensive search over large discrete model spaces. 

Our approach is not limited to univariate regression and can be extended to the varying coefficient model \citep{hastie1993varying}, density estimation via log-spline models \citep{kooperberg1991study} and additive models \citep{hastie1986generalized}, among others. Further details are provided in Section \ref{sec:sim}. In the additive regression context, the proposed approach performs comparably to state-of-the-art procedures like the {\it Sparse Additive Model} (SpAM) of \cite{ravikumar2009sparse} and the {\it High-dimensional Generalized Additive Model} (HGAM) of \cite{meier2009high}.


We provide theoretical justification for the method by showing an adaptive property of the approach. Specifically, we show that the posterior contracts \citep{ghosal2000convergence} at the parametric rate if the true function belongs to the pre-designated subspace, and contracts at the optimal rate for $\alpha$-smooth functions otherwise. In other words, our approach adapts to the parametric shape of the unknown function while allowing deviations from the parametric shape in a nonparametric fashion. In addition, we describe a model selection procedure obtained by thresholding the shrinkage factor,  and then demonstrate its consistency.

 \section{Preliminaries}\label{sec:pre}
We begin by introducing some notation. For $\alpha > 0$, let $\lfloor \alpha \rfloor$ denote the largest integer smaller than or equal to $\alpha$ and  $\ceil{\alpha}$ denote the smallest integer larger than or equal to $\alpha$. Let $C^{\alpha}[0,1]$ denote the H{\"o}lder class of $\alpha$ smooth functions on $[0,1]$ that have continuously differentiable derivatives up to order $\lfloor \alpha \rfloor$, with the $\lfloor \alpha \rfloor$th order derivative being Lipschitz continuous of order $\alpha - \lfloor \alpha \rfloor$. 
For a vector $x \in \mb R^d$, let $\norm{x}$ denote its Euclidean norm. For a function $g:[0,1] \to \mb R$ and points $x_1, \ldots, x_n \in [0,1]$, let $\norm{g}_{2,n}^2 = n^{-1} \sum_{i=1}^n g^2(x_i)$; we shall refer to $\norm{\cdot}_{2,n}$ as the empirical $L_2$ norm. For an $m \times d$ matrix $A$ with $m > d$ and $\mbox{rk}(A) = d$, let $\mathfrak{L}(A)=\{ A\beta : \beta \in \mb R^d \}$ denote the column space of $A$, which is a $d$-dimensional subspace of $\mb R^m$. Let $\Proj_A = A(A^{\T} A)^{-1} A^{\T}$ denote the projection matrix on $\mathfrak{L}(A)$.

\section{The functional horseshoe prior}\label{sec:main}

In the nonparametric regression model in \eqref{eq:normal_mean}, we model the unknown function $f$ as spanned by a set of pre-specified basis functions $\{\phi_j\}_{1\leq j\leq k_n}$ as follows:
\begin{eqnarray}\label{eq:ff}
f(x) = \sum_{j=1}^{k_n} \beta_j \phi_j(x).
\end{eqnarray}                                    
We work with the B-spline basis \citep{de2001practical} for illustrative purposes here. However, the methodology trivially  generalizes to a larger class of basis functions. A detailed description of the B-spline basis is provided in Section \ref{sec:suppBbasis} in the supplementary material. Let $\beta = (\beta_1, \ldots, \beta_{k_n})^{\T}$ denote the vector of basis coefficients and let $\Phi = \{\phi_j(X_i)\}_{1 \le i \le n, 1 \le j \le k_n}$ denote the $n \times k_n$ matrix of basis functions evaluated at the observed covariates. Model \eqref{eq:normal_mean} can then be expressed as 
\begin{eqnarray}\label{eq:basis_exp}
Y \mid \beta \sim \mbox{N}(\Phi \beta, \sigma^2 \mr I_n). 
\end{eqnarray}
%
A standard choice for a prior on $\beta$ is a $g$-prior, $\beta \sim \mbox{N}(0, g (\Phi^{\T} \Phi)^{-1})$\citep{zellner1986}. These priors are commonly used in linear models because they incorporate the correlation structure of the covariates inside the prior variance. The posterior mean of $\beta$ under a $g$-prior can be expressed as $\{ 1 - 1/(1+g)\} \widehat{\beta}$, where $\widehat{\beta} = \Proj_{\Phi} Y$ is the maximum likelihood estimate of $\beta$. Thus, the posterior mean shrinks the maximum likelihood estimator towards zero, with the amount of shrinkage controlled by the parameter $g$.
 \citet{bontemps2011bernstein} studied asymptotic properties of the resulting posterior by providing bounds on the total variation distance between the posterior distribution and a Gaussian distribution centered at the maximum likelihood estimator with the inverse Fisher information matrix as covariance. In \citet{bontemps2011bernstein}, the $g$ parameter was fixed {\em a priori} depending on the sample size $n$ and the error variance $\sigma^2$. In particular, the results of \citet{bontemps2011bernstein} imply minimax optimal posterior convergence for $\alpha$-smooth functions. In related work, \citet{ghosal2007convergence} established minimax optimality with isotropic Gaussian priors on $\beta$. 

 Our goal is to define a broader class of shrinkage priors on $\beta$ that facilitate shrinkage towards a {\em null subspace} that is fixed in advance, rather than shrinkage towards the origin or any other fixed {a priori} guess $\beta_0$. For example, if we have {\em a priori} belief that the function is likely to attain a linear shape, then we would like to impose shrinkage towards the class of linear functions. In general, our methodology allows shrinkage towards any null subspace spanned by the columns of a null regressor matrix $\Phi_0$, with $d_0 = \mbox{rank}(\Phi_0)$ equal to the dimension of the null space. For example in the linear case, we define the null space as $\mathfrak{L}(\Phi_0)$ with $\Phi_0 = \{{\bf 1},\bx\} \in \mb R^{n \times 2}$, where ${\bf 1}$ is a $n \times 1$ vector of ones and $d_0 = 2$. Shrinkage towards quadratic, or more generally polynomial, regression models is achieved similarly.  

With the above notation, we define the fHS prior through the following conditional specification: 
\begin{eqnarray}
\pi(\beta\mid\tau) &\propto& (\tau^2)^{-(k_n-d_0)/2}\exp\left\{  -\frac{1}{2\sigma^2\tau^2} \beta^\T\Phi^\T\IP\Phi\beta   \right\}, \label{eq:beta_pr} \\
\pi(\tau) &\propto& \frac{(\tau^2)^{b-1/2 }}{(1+\tau^2)^{(a+b)}} \ind_{(0, \infty)}(\tau), \label{eq:hyper}
\end{eqnarray}
where $a, b > 0$. Recall that $\Proj_0 = \Phi_0(\Phi_0^\T\Phi_0)^{-1}\Phi_0^\T$ denotes the projection matrix of $\Phi_0$. 

When $\Phi_0 = 0$, \eqref{eq:beta_pr} is equivalent to a $g$-prior with $g = \tau^2$. The key additional feature  in our proposed prior is the introduction of  the quantity $\IP$ in the exponent, which enables shrinkage towards subspaces rather than single points. Although the proposed prior may be singular, it follows from  subsequent results that the joint posterior on $(\beta, \tau^2)$ is proper. Note that the prior on the scale parameter $\tau$ follows a half-Cauchy distribution when $a = b = 1/2$. Half-Cauchy priors have been recommended as a default prior choice for global scale parameters in the linear regression framework \citep{polson2012half}.  Using the reparameterization $\omega = 1/(1 + \tau^2)$, the prior in \eqref{eq:hyper} can be interpreted as the prior induced on $\tau$ through  a $\mbox{Beta}(a, b)$ prior on $\omega$. We work in the $\omega$ parameterization for reasons to be evident shortly. 

Exploiting the conditional Gaussian specification, the conditional posterior of $\beta$ is also Gaussian, and can be expressed as 
\begin{align}\label{eq:beta_post}
\beta \mid Y, \omega \sim \mbox{N}( \wt{\beta}_{\omega}, \wt\Sigma_{\omega}),  
\end{align}
where
{\small \begin{align}\label{eq:beta_mv}
\wt{\beta}_{\omega} =  \left(\Phi^{\T} \Phi+ \frac{\omega}{1-\omega}\Phi^\T\IP\Phi\right)^{-1} \Phi^{\T} Y, \quad \wt\Sigma_{\omega} = \sigma^2  \left(\Phi^{\T} \Phi+\frac{\omega}{1-\omega}\Phi^\T\IP\Phi \right)^{-1}. 
\end{align}}
We now state a lemma which delineates the role of $\omega$ as the parameter controlling the shrinkage. 
\begin{lemma}\label{lem:postm}
Suppose that $\mathfrak{L}(\Phi_0) \subsetneq \mathfrak{L}(\Phi)$. Then,
\be
\bbE\left[\Phi\beta \mid Y,\omega  \right] = \Phi\wt\beta_\omega= (1-\omega)\Proj_{\Phi} Y + \omega \Proj_0 Y,
\ee
where $\Proj_\Phi$ is the projection matrix of $\Phi$.
\end{lemma}
This lemma shows that the conditional posterior mean of the regression function given $\omega$ is a convex combination of the classical B-spline estimator $\Proj_{\Phi} Y$ and the parametric estimator $\Proj_0 Y$. The parameter $\omega\in(0,1)$ controls the shrinkage effect; the closer $\omega$ is to $1$, the greater the shrinkage towards the parametric estimator. We learn the parameter $\omega$ from the data with a $\mbox{Beta}(a, b)$ prior on $\omega$. The hyperparameter $b < 1$ controls the amount of prior mass near one. 

 Figure \ref{fig:beta} illustrates the connection between  the choice of the hyperparameters $a$ and $b$ and the shrinkage behavior of the prior. The first and the second column in Figure \ref{fig:beta}, with  $a$ fixed at $1/2$, shows that the prior density of $\omega$ increasingly concentrates near $1$ as $b$ decreases from $1/2$ to $1/10$. The third column in Figure \ref{fig:beta} depicts the prior probability that $\omega>0.95$ and $\omega<0.05$. Clearly, as $b$ decreases, the amount of prior mass around one increases, which results in stronger shrinkage towards the parametric estimator. In particular, when $a= b = 1/2$, the resulting functional ``HS" prior density derives its name from the shape of the prior on $\omega$ \citep{carvalho2010horseshoe}. 
\begin{figure} 
\includegraphics[height=4.5cm, width=16cm]{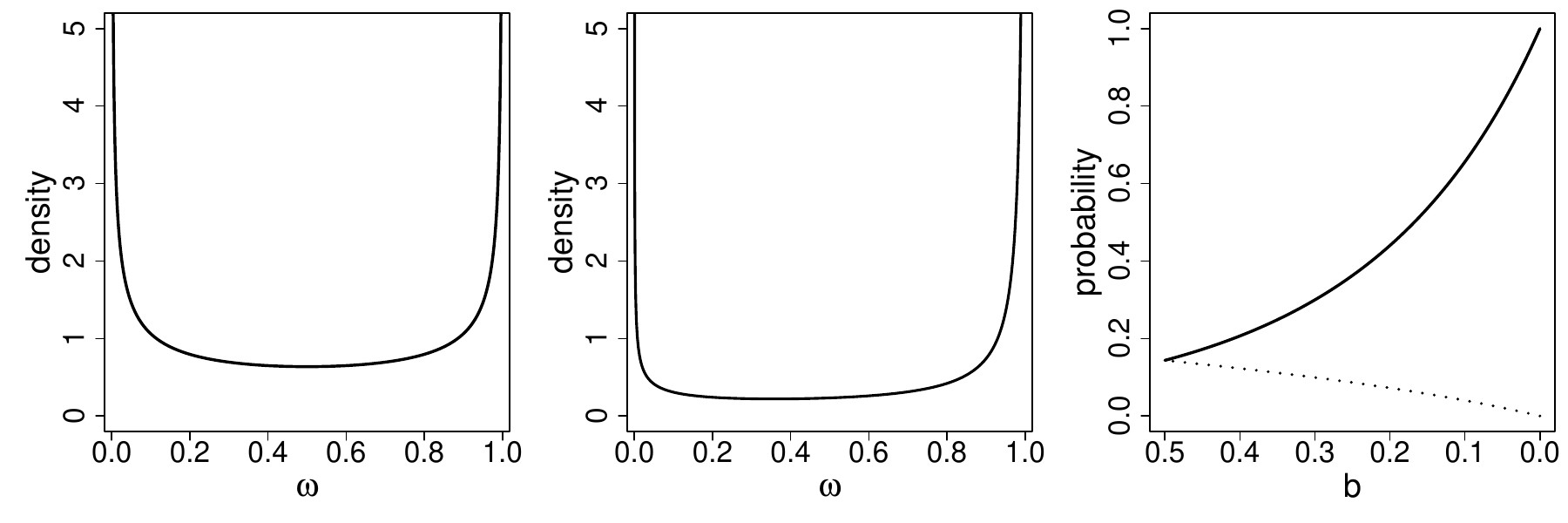}
\caption{The first two columns illustrate the prior density function of $\omega$ with different hyperparameters $(a,b)$:  $(1/2,1/2)$ for the first column and $(1/2,1/10)$ for the second column. The third column shows the prior probability that $\omega>0.95$ (solid line) and $\omega<0.05$ (dotted line) for varying $b$ and a fixed $a=1/2$.} 
\label{fig:beta}
\end{figure}

When $\mathfrak{L}(\Phi_0) \subsetneq \mathfrak{L}(\Phi)$, we can orthogonally decompose $\Proj_\Phi = \Proj_1 + \Proj_0$, where the columns of $\Proj_1$ are orthogonal to the columns of $\Proj_0$, i.e., $\Proj_1^{\T} \Proj_0 = 0$. For $\mathfrak{L}(\Phi_0)\subsetneq 
\mathfrak{L}(\Phi)$, this follows because we can use Gram-Schmidt orthogonalization to create $\wt{\Phi} = [\Phi_0; \Phi_1]$ of the same dimension as $\Phi$ with $\Phi_1^{\T} \Phi_0 = 0$ and $\mathfrak{L}(\Phi) = \mathfrak{L}(\wt{\Phi})$. Let $\Proj_{1}$ denote the projection matrix on $\mathfrak{L}(\Phi_1)$. Simple algebra shows that 
\begin{align}
 & \pi(\omega \mid Y) = \int \pi(\omega, \beta \mid Y) d \beta =   \frac{\pi(\omega)}{m(Y)} \int f(Y \mid \beta, \omega) \pi(\beta \mid \omega) d \beta \notag \\
& = \omega^{a+(k_n-d_0)/2-1}(1-\omega)^{b-1}\exp\{-H_n \omega \}/m(Y), \label{eq:omega}
\end{align}
where $H_n = Y^{\T} \Proj_{1} Y/(2\sigma^2)$ and $m(Y) = \int^1_0 \omega^{a+(k_n-d_0)/2-1}(1-\omega)^{b-1}\exp\left\{  -H_n\omega \right\} d\omega$.

To investigate the asymptotic behavior of the resulting posterior, it is crucial to find tight two-sided bounds on $m(Y)$. Such bounds are specified in Lemma \ref{lem:C}. 

\begin{lemma}{(Bounds on the normalizing constant)}\label{lem:C} Let $A_n$ and $B_n$ be  arbitrary sequences satisfying $A_n\to \infty$ as $n \to \infty$ and $B_n=O(1)$. Define $t_n = \int_0^1 \omega^{A_n-1}(1-\omega)^{B_n-1}\exp\{-H_n\omega\}d\omega$.  Then,
\be
\frac{\Gamma(A_n)\Gamma(B_n)}{\Gamma(A_n+B_n)}\exp\{-H_n\}(1+Q_n^L)  \leq t_n \leq \frac{\Gamma(A_n)\Gamma(B_n)}{\Gamma(A_n+B_n)}\exp\{-H_n\}(1+Q_n^U), 
\ee
where,  
\be
Q_n^U &=&  \frac{B_n}{A_n+B_n}\exp( H_n),\\
Q_n^L &=&  \frac{B_nH_n}{A_n+B_n} + \frac{DB_n(B_n+T_n)^{-A_n}}{(A_n+B_n)^{3/2}}\left(\exp\{H_n\} - 1 -H_n - (T_n+2)^{-1/2}\right)_+, 
\ee
where $T_n=\max\{A_n^2,  3\ceil{H_n}\}$ and   $D$ is some positive constant.
\end{lemma}
By setting $A_n=a+k_n/2$ and $B_n = b$, Lemma \ref{lem:C} shows that the magnitude of the normalizing constant $m(Y)$ in \eqref{eq:omega}  is determined by an interplay between the relative sizes of $b$ and $\exp(H_n)$. When $b$ is small enough so that  $b\exp(H_n)\approx 0$, it follows that $m(Y) \approx \mbox{Be}(a+k_n/2,b) \exp(-H_n)$, where $\mbox{Be}(\cdot,\cdot)$ denotes the beta function. Otherwise, ignoring polynomial terms, $m(Y) \approx Be(a+k_n/2,b)b$. This asymptotic behavior of $m(Y)$ is central to identifying the posterior contraction rate of the fHS prior. We also note that  the magnitude of $a$ asymptotically does not affect the strength of shrinkage for large $n$ as long as $a$ is a fixed constant, since  the prior contribution $\omega^{a-1}$ is dominated by the likelihood contribution $\omega^{k_n/2}$.

\subsection{  Posterior concentration rate}\label{sec:post}

We assume a set of standard regularity conditions that have been used by others (\cite{zhou1998local}, \cite{claeskens2009asymptotic}) to prove minimax optimality of B-spline estimators. These regularity conditions are described in Section \ref{sec:suppBbasis} in the supplementary material.
Under the regularity conditions, \cite{zhou1998local} showed that the mean square error of the B-spline estimator $\Proj_\Phi Y$ achieves the minimax optimal rate. If the true function $f_0 \in C^{\alpha}[0, 1]$ is $\alpha$-smooth and the number of basis functions $k_n \asymp n^{1/(2 \alpha + 1)}$, then they  showed that
\begin{eqnarray}
\bbE_0\left[ \norm{\Proj_\Phi Y - F_0}_{2,n}^2\right] = O\left(n^{-2\alpha/(1+2\alpha)}\right),\label{eq:classic}
\end{eqnarray}
where $\bbE_0(\cdot)$ represents an expectation with respect to the true data generating distribution of $Y$. We now state our main result on the posterior contraction rate of the fHS prior. 
\begin{theorem}
\label{theo:post}
Consider the model \eqref{eq:normal_mean} equipped with the fHS prior \eqref{eq:beta_pr}-\eqref{eq:hyper}. Assume $\mathfrak{L}(\Phi_0) \subsetneq \mathfrak{L}(\Phi)$. Further assume that for some integer $\alpha \ge 1$, the true regression function $f_0\in C^\alpha[0,1]$ and the B-spline basis functions $\Phi$ are constructed with  $k_n-\lfloor\alpha\rfloor$ knots and $\lfloor\alpha\rfloor-1$ degree, where $k_n \asymp n^{1/(1+2\alpha)}$. Suppose that 
the prior hyperparameters $a$ and $b$ in \eqref{eq:hyper} satisfy $a \in (\delta, 1 - \delta)$ for some constant $\delta \in (0, 1/2)$, and $k_n\log k_n\prec -\log b \prec (n k_n)^{1/2}$. Then, for any diverging sequence $\zeta_n$, 
$\bbE_0[ P\{\norm{\Phi\beta - F_0}_{2,n}
>M_n(f_0)^{1/2} \mid Y \} ] = o(1)$,
where 
\begin{eqnarray*}
	M_n(f_0) =\begin{cases}
    \zeta_n n^{-1},\mbox{if $F_0 \in \mathfrak{L}(\Phi_0)$}\\
	\zeta_n n^{-2\alpha/(1+2\alpha)}\log n,\mbox{ if $F_0^\T\IP F_0 \asymp n$}.
   \end{cases}
\end{eqnarray*}
\end{theorem}
Theorem \ref{theo:post} exhibits an adaptive property of the fHS prior. If the true function is $\alpha$-smooth, then the posterior contracts around the true function at the near minimax rate of $n^{-\alpha/(2\alpha + 1)} \log n$. However, if the true function $f_0$ belongs to the finite dimensional subspace $\mathfrak{L}(\Phi_0)$, then the posterior contracts around $f_0$ in the empirical $L_2$ norm at the parametric $n^{-1/2}$ rate. We note that the bound $k_n \log k_n \prec -\log b \prec (n k_n)^{1/2}$ is a key to the adaptivity of the posterior, since the strength of the shrinkage  towards $\mathfrak{L}(\Phi_0)$ is controlled by $b$. If $-\log b \prec k_n \log k_n$, then the shrinkage  towards $\mathfrak{L}(\Phi_0)$ is too weak to achieve the parametric rate when $F_0\in\mathfrak{L}(\Phi_0)$. On the other hand, if  $-\log b \succ (nk_n)^{1/2}$, the resulting posterior distribution strongly concentrates around  $\mathfrak{L}(\Phi_0)$ and fails to attain the optimal nonparametric rate of posterior contraction when $F_0\not\in\mathfrak{L}(\Phi_0)$.   
 
  We ignore the subspace of functions such that $\{F \in \mathbb{R}^n: F^\T \IP F = o(n),\:\:F\not\in \mathfrak{L}(\Phi_0)\}$ and  only focus on functions that can be strictly separated from the null space $\mathfrak{L}(\Phi_0)$. However, we acknowledge that it would be useful to illustrate the shrinkage behavior when the regression function $f$  approaches the null space under the condition that $\lim_{n \to \infty} F^\T \IP F/n = 0$.

\subsection{Model selection procedure and its consistency}\label{sec:modelsel}
In this section, we illustrate a model selection procedure based on the fHS priors and examine their theoretical consistency. As mentioned in Lemma \ref{lem:postm}, $\omega$ can be interpreted as the amount of weight that the posterior mean for function $F$ places on the parametric estimator $Q_0 Y$. Due to this fact, it is natural to consider a model selection procedure by thresholding the posterior mean of $\omega$  analogous to the model selection procedure considered in  \cite{carvalho2010horseshoe} for the standard HS prior. Since  a posterior mean of $\omega$ that is larger than $1/2$ indicates that more weight is imposed on the parametric estimator compared to the amount of the weight on the nonparametric estimator, it is natural to select the parametric model when $E(\omega\mid Y)>1/2$.  


  The asymptotic properties of such a thresholding based model selection procedure depends on the behavior of $\omega$ {\em a posteriori}.  The following theorem states the posterior convergence rate of $\omega$ when the true function belongs to the parametric or nonparametric family.     
                                                                                                                                                                                                                                                                                                                                                                                                                                                                                                                                                                                                                                                                                                                                                                                                                                                                                                                                                                                                                                                                                                                                                                                                                                                                                                                                                                                                                                                                                                                                                                                                                                                                                                                                                                                                                                                                                                                                                                                                                                                                                                                                                                                                                                                                                                                                                                                                                                                                                                                                                                                                                                                                                                                                                                                                                                                                                                                                                                                                                                                                                                                                                                                                                                                                                                                                                                                                                                                                                                                                                                                                                                                                                                                                                                                                                                                                                                                                                                                                                                                                                                                                                                                                                                                                                                                                                                                                                                                                                                                                                                                                                                                                                                                                                                                                                                                                                                                                                                                                                                                                                                                                                                                                                                                                                                                                                                                                                                                                                                                                                                                                                                                                                                                                                                                                                                                         
\begin{theorem}{\textnormal{ (posterior convergence rate of $\omega$)}}\label{theo:modelsel}
Assume conditions from Theorem \ref{theo:post} hold. Then, for any diverging sequence $\zeta_n$ and any constant $\epsilon_0>0$, $
E_0\left[ P(\omega < 1- \zeta_n S_{0,n}\mid Y) \right]= o(1)$ if $F_0\in \mathfrak{L}(\Phi_0)$,	and
 $E_0\left[ P(\omega >  \zeta_n S_{1,n}\mid Y) \right] = o(1)$  if $F_0^\T\IP F_0 \asymp n $, where $S_{0,n} = k_n^{-1}b^{1-\epsilon_0}$ and $S_{1,n}= (-\log b)/n$.
 \end{theorem}
 
 Theorem \ref{theo:modelsel} indicates that when the true function is parametric, the posterior distribution of $\omega$ contracts towards $1$ at a rate of at least $k_n^{-1}b^{1-\epsilon_0}$ for any $\epsilon_0>0$. On the other hand, when the true function is strictly separated from the class of parametric functions, i.e., $F_0^\T\IP F_0 \asymp n $, the posterior distribution of $\omega$ converges to zero at a rate of $-\log b/n$. By the condition  $k_n\log k_n\prec -\log b \prec (n k_n)^{1/2}$ in Theorem \ref{theo:post}, both $k_n^{-1}b^{1-\epsilon_0}$ and $-\log b /n$ converge to zero. These results guarantee the consistency of the model selection procedure  based on thresholding $E(\omega\mid Y)$ by any value in $(0,1)$.
 

\section{Examples for the univariate case}\label{sec:sim}
In this section, we consider some applications of the fHS prior for several nonparametric models:  
\begin{eqnarray}
\mbox{(i) simple regression model:}&& Y_i = f(x_i) + \epsilon_i \label{eq:reg}\\
\mbox{(ii) varying coefficient model:}&\:& Y_i = t_if(x_i) +\epsilon_i\label{eq:vc}\\
\mbox{(iii) density function estimation:}&\:& p(Y_i) = \frac{\exp\{ f(Y_i) \} }{\int \exp\{ f(t) \}dt }, \label{eq:density}
\end{eqnarray}
In cases (i) and (ii), we assume that $\epsilon_i\overset{i.i.d}\sim\Gauss(0,\sigma^2)$ for $i=1,\dots,n$. In (ii), $t_i$ and $x_i$ are covariates for $i=1,\dots,n$. In (iii), $p(\cdot)$ is the unknown density function of $Y$. The varying coefficient model \citep{hastie1993varying} in \eqref{eq:vc} reduces to a linear model when the coefficient function $f$ is constant, and the density function $p$ is Gaussian when the log-density function $f$ is quadratic in the log-spline model \eqref{eq:density} \citep{kooperberg1991study}. These facts motivate the use of the fHS prior in these examples to shrink towards the respective parametric alternatives. 

\begin{figure}
\includegraphics[height=4.5cm, width=15cm]{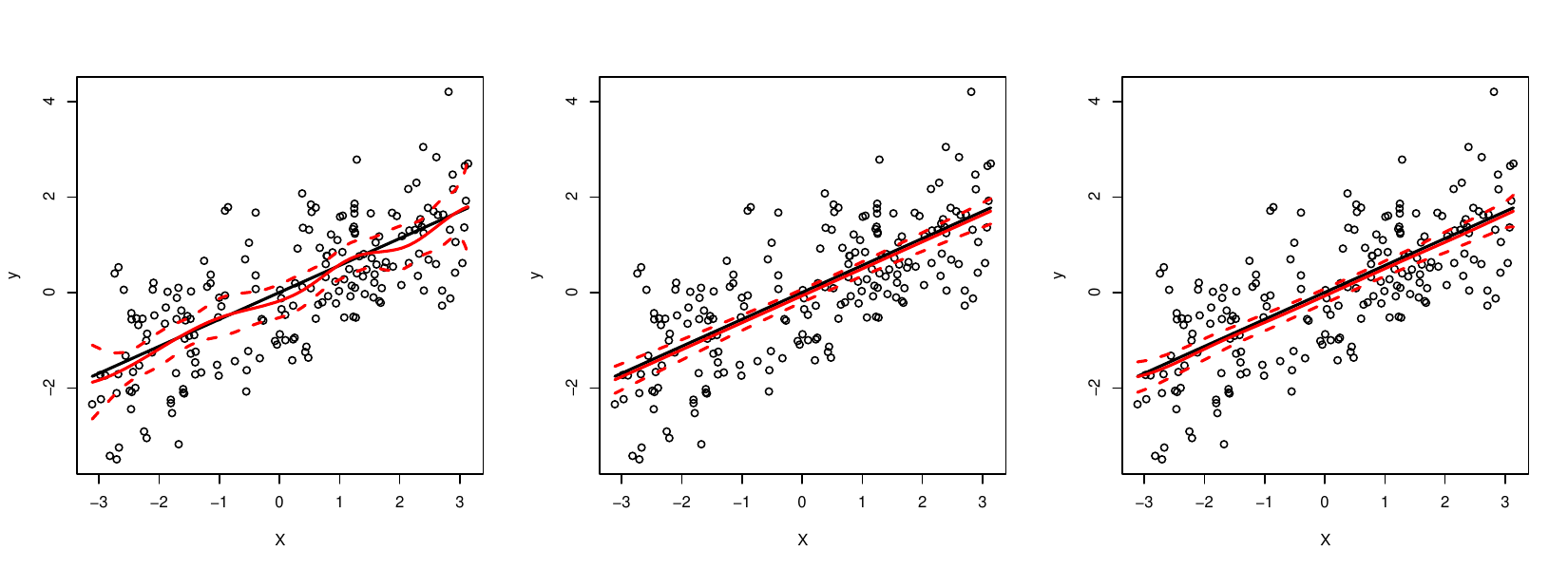}
\includegraphics[height=4.5cm, width=15cm]{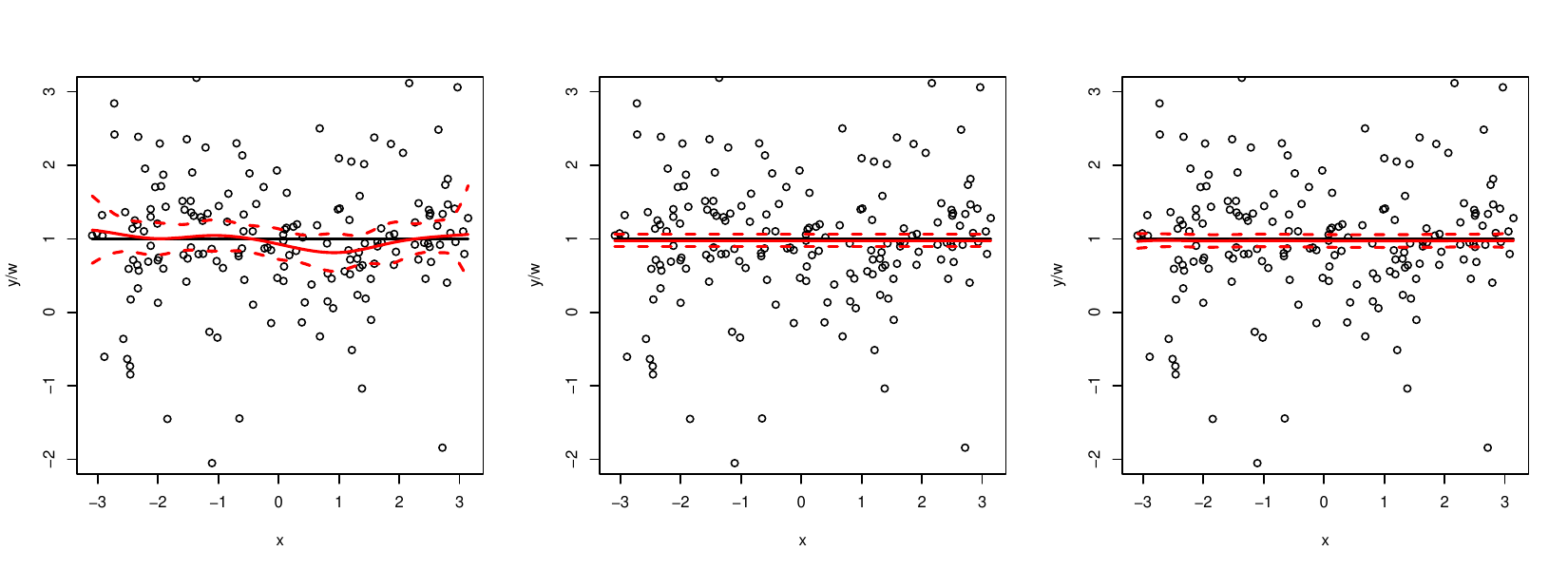}
\includegraphics[height=4.5cm, width=15cm]{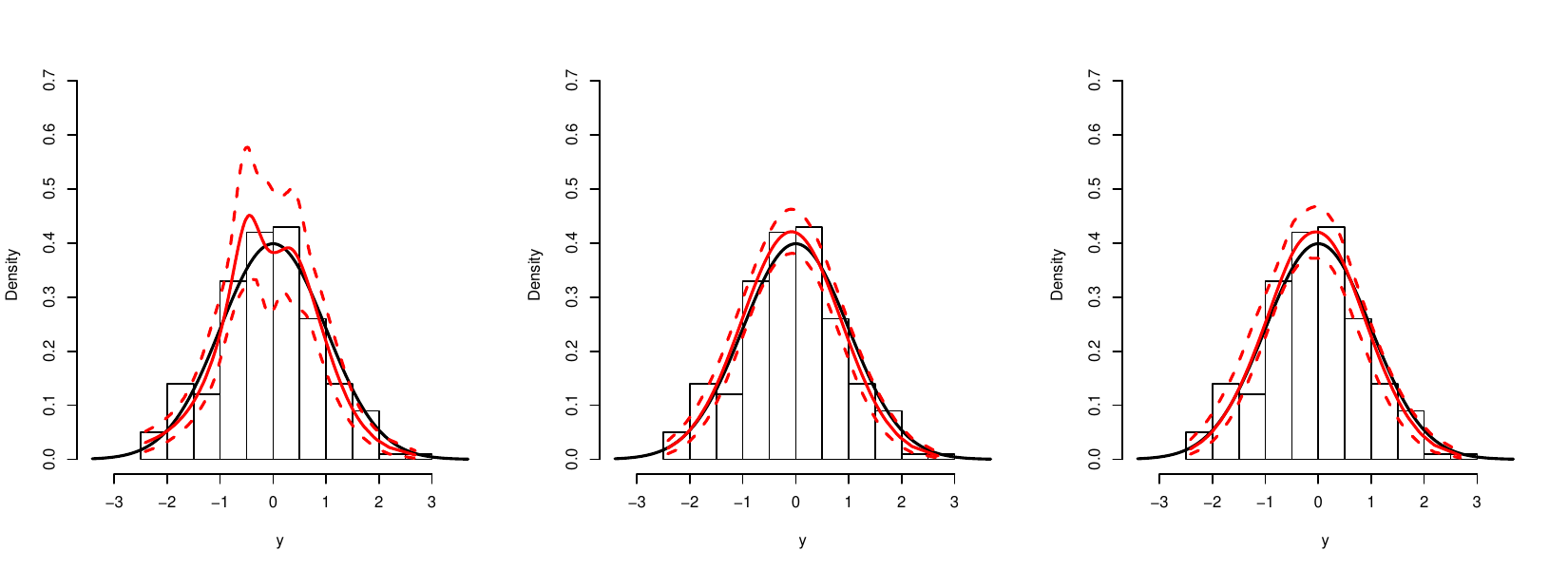}
\caption{Examples when the underlying true functions are  parametric. Posterior mean of each procedure (red solid), its 95\% pointwise credible bands (red dashed), and the true function (black solid) from a single example with $n=200$ for each model. The top row is for the simple regression model; the second row is for the varying coefficient model; the last row is for density estimation. The Bayesian B-spline procedure, the Bayesian parametric model procedure, and fHS priors are illustrated in the first, second, and third columns, respectively.  } 
\label{fig:example}
\end{figure}

\begin{figure}
\includegraphics[height=4.5cm, width=15cm]{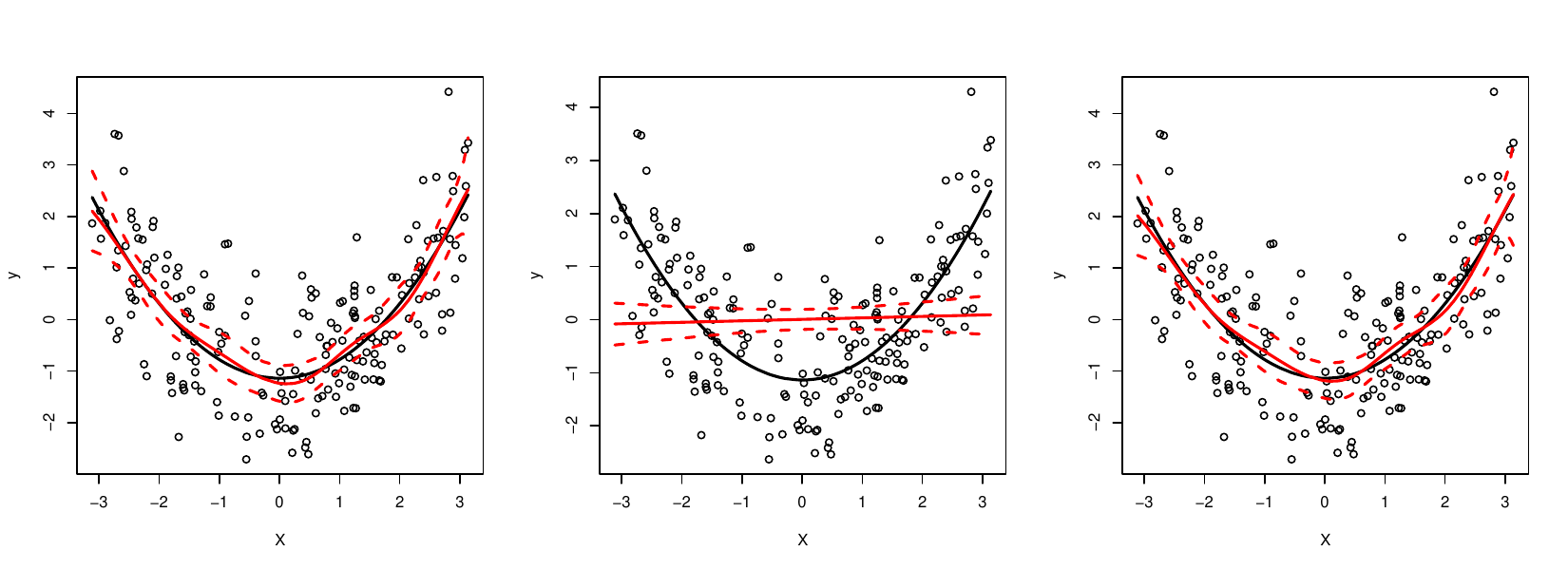}
\includegraphics[height=4.5cm, width=15cm]{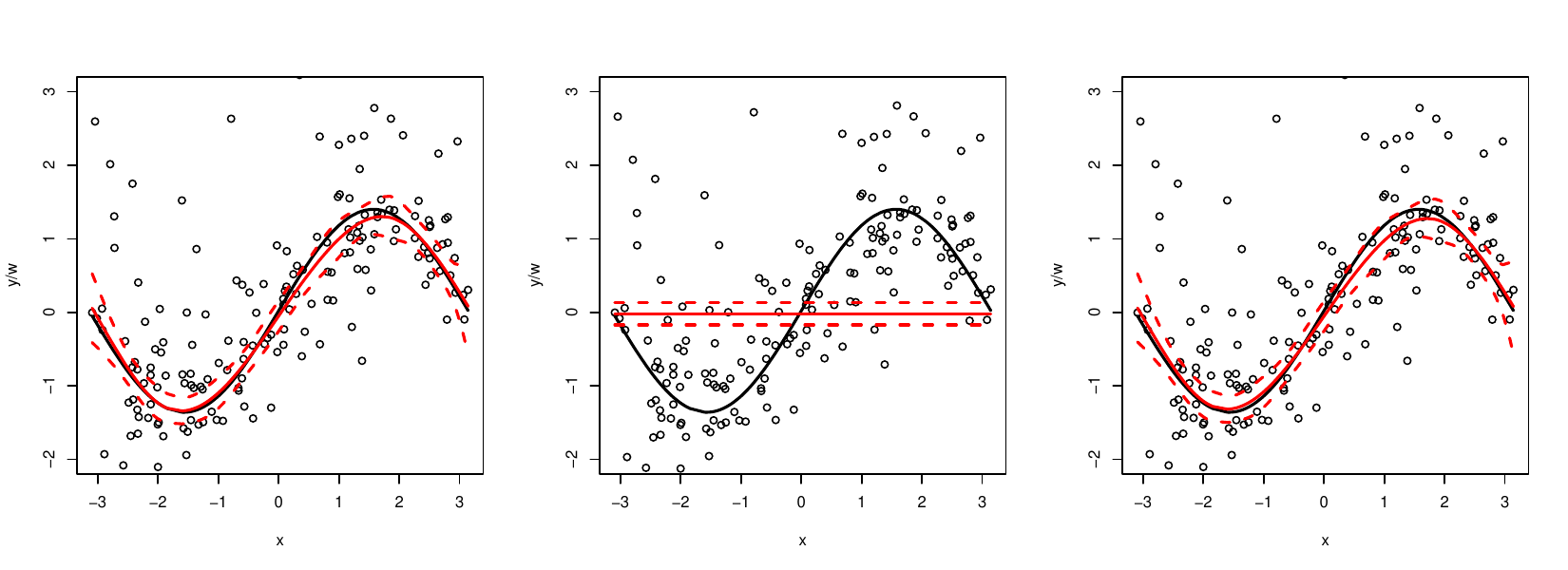}
\includegraphics[height=4.5cm, width=15cm]{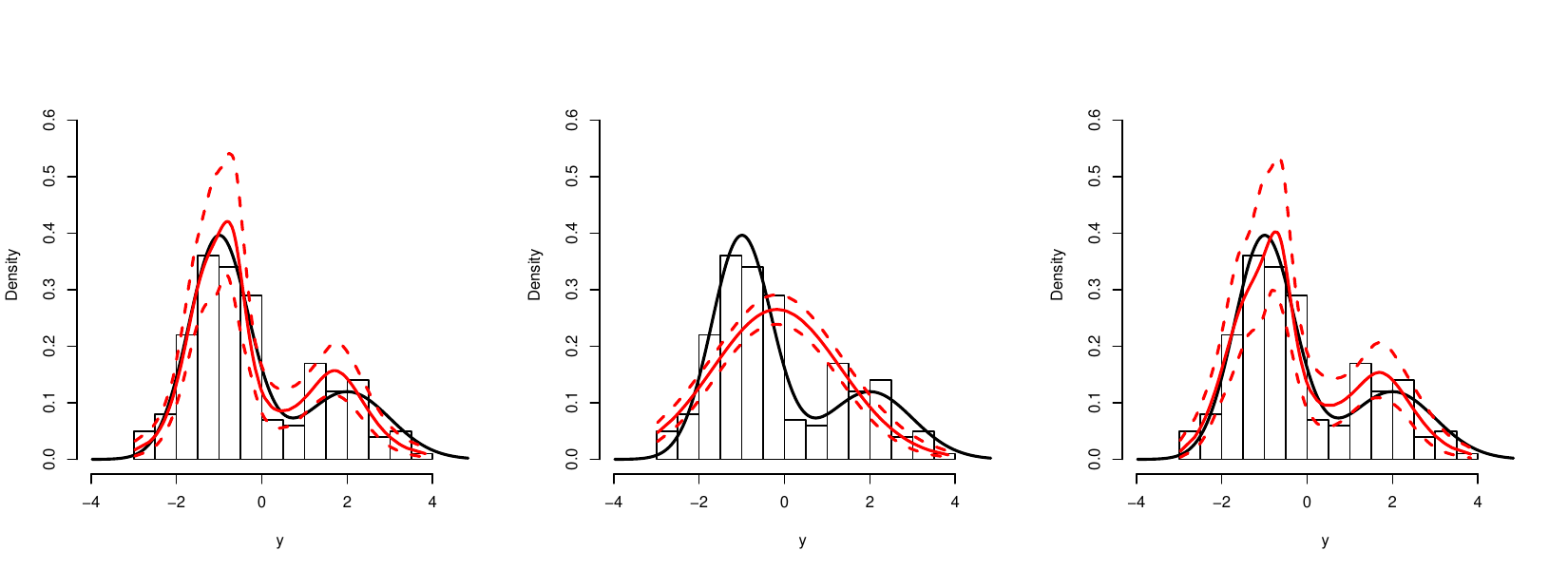}
\caption{ Examples when the underlying true functions are nonparametric. The description of the figures are provided in the caption of Figure \ref{fig:example}. } 
\label{fig:example2}
\end{figure}

Before providing a detailed simulation study, we illustrate in Figures \ref{fig:example} and \ref{fig:example2} what we generally expect from the fHS procedure. Figure \ref{fig:example} depicts the point estimate (posterior mean) and pointwise 95\% credible bands for the unknown function $f$ for a single data set for each of the three examples when the true function belongs to the parametric class. That is, a linear function in \eqref{eq:reg}, a constant function in \eqref{eq:vc}, and a quadratic function in \eqref{eq:density}. Figure \ref{fig:example2} depicts the corresponding estimates when the data generating function does not fall into  the assumed parametric class. It is evident from Figure \ref{fig:example} that when the parametric assumptions are met, the fHS prior performs similarly to the parametric model. This fact empirically corroborates our findings in Theorem \ref{theo:post} that the posterior contracts at a nearly parametric rate when the parametric assumptions are met. It is also evident that the fHS procedure automatically adapts to deviations from the parametric assumptions in Figure \ref{fig:example2}, again confirming the conclusion of Theorem \ref{theo:post}. That is, when the true function is well-separated from the parametric class, the posterior concentrates at a near optimal minimax rate. We reiterate that the same hyperparameters $a=1/2$ and $b=\exp\{-k_n\log n /2\}$ for the fHS prior were used in the examples in Figure \ref{fig:example} and Figure \ref{fig:example2}. 

We now provide the details of a replicated study for the simple regression model. The details for the varying coefficient model and the log-density model are provided in Section B of the supplementary documents, with the overall message consistent across the different problems. We generated the covariates independently from a uniform distribution between $-\pi$ and $\pi$ and set the error variance $\sigma^2 = 1$. We considered three parametric choices for $f$. These  include  linear, quadratic, and sinusoidal functions. We standardized the true function so as to obtain a signal-to-noise ratio of $1.0$.

To shrink the regression function in \eqref{eq:reg} towards linear subspaces for the simple regression model, we set $\Phi_0= \{{\bf 1}, \bx \}$ in the fHS prior \eqref{eq:beta_pr} ($\Phi_0= \{{\bf 1} \}$ for the varying coefficient model and $\Phi_0= \{{\bf 1}, \bx,\bx^2 \}$ for the log-density model).  An inverse-gamma prior with parameters $(1/100,1/100)$ was imposed on $\sigma^2$ for the fHS prior, and we set $b = \exp\{ -k_n\log n/2 \}$ to satisfy the conditions of Theorem \ref{theo:post}.  We  arbitrarily set $a=1/2$. We consider the number of basis functions $k_n\in\{5,8,11,35\}$. In particular, the choice $k_n=35$ was empirically recommended when $n>140$ in \cite{ruppert2003semiparametric}.   

To compare the fHS prior to the standard horseshoe (HS) prior, we considered a decomposition $F = F_0 + F_1$, where $F_0$ is the parametric function and $F_1 = \Phi\beta$ is the nonparametric component modeled by the B-spline basis functions. The parametric form $F_0$ is set to be linear. For a performance comparison to our procedure, we imposed the standard HS prior on the coefficients of the B-spline basis functions to encourage shrinkage of the nonparametric part towards zero in a different fashion than the fHS method. 

  We also considered a penalized spline procedure for the performance comparisons. The object function of the penalized likelihood can be expressed as $\norm{Y-\Phi\beta}^2_2 + \lambda \beta^T \Sigma \beta$,  where $\Sigma$ is a $k_n$ by $k_n$ matrix with $\Sigma_{jk} = \int \phi_j^{''}(t)\phi_k^{''}(t)dt$ for $j,k=1,\dots,k_n$. The smoothness parameter $\lambda$ was chosen by generalized cross-validation \citep{golub1979generalized}.  
\begin{sidewaystable*}
	
\centering
\resizebox{19cm}{!}{%

\begin{tabular}{|c|cccc|cccc|cccc|}
\hline
Truth&  \multicolumn{4}{|c|}{Linear} & \multicolumn{4}{|c|}{Quadratic} & \multicolumn{4}{|c|}{Sine} \\
  \hline
 $n=200$ & $k_n = 5$ & $k_n = 8$ & $k_n = 11$ & $k_n=35$ & $k_n = 5$ & $k_n = 8$ & $k_n = 11$ & $k_n=35$ &$k_n = 5$ & $k_n = 8$ & $k_n = 11$ & $k_n=35$ \\ 
  \hline
Oracle & 0.918(0.08) &  &  &  & 2.961(0.13) & 3.593(0.16) & 5.052(0.20) & 17.088(0.35) & 2.602(0.13) & 3.568(0.16) & 5.049(0.20) & 17.088(0.35) \\ 
  PenSpline & 1.563(0.11)  & 2.536(0.13) & 3.555(0.29) & 12.939(0.29) & {\bf2.187}(0.13) & {\bf2.563}(0.13) & {\bf3.557}(0.15) & 12.983(0.29) & 3.344(0.13) & {\bf2.618}(0.13) & {\bf3.660}(0.15) & 13.745(0.29) \\ 
  HS & 1.233(0.09) & 1.591(0.10) & 2.030(0.11) & 5.436(0.19) & 3.243(0.13) & 3.278(0.15) & 4.646(0.18) & {\bf11.203}(0.26) & {\bf2.191}(0.12) & 3.280(0.15) & 3.968(0.16) & {\bf8.847}(0.23) \\ 
  fHS1 & 1.109(0.08) & 0.934(0.08) & 0.926(0.08) & 0.922(0.08) & 4.237(0.11) & 3.627(0.17) & 5.031(0.21) & 15.162(0.35) & 2.701(0.13) & 3.640(0.16) & 5.017(0.21) & 16.579(0.91) \\ 
  fHS2 & {\bf1.101}(0.08) & {\bf0.933}(0.08) & {\bf0.925}(0.08) & {\bf0.922}(0.08) & 5.116(0.11) & 3.627(0.17) & 5.031(0.21) & 15.162(0.35) & 2.702(0.13) & 3.641(0.16) & 5.018(0.21) & 16.579(0.91) \\ 
  fHS3 & {\bf1.101}(0.08) & {\bf0.933}(0.08) & {\bf0.925}(0.08) & {\bf0.922}(0.08) & 5.381(0.15) & 3.627(0.17) & 5.031(0.21) & 15.162(0.35) & 2.702(0.13) & 3.641(0.16) & 5.018(0.21) & 16.579(0.91) \\ 
\hline
  $n=500$ & $k_n = 5$ & $k_n = 8$ & $k_n = 11$ & $k_n=35$ & $k_n = 5$ & $k_n = 8$ & $k_n = 11$ & $k_n=35$ &$k_n = 5$ & $k_n = 8$ & $k_n = 11$ & $k_n=35$ \\ 
  \hline
  Oracle & 0.425(0.04) &  &  &  & 1.626(0.06) & 1.559(0.07) & 2.136(0.09) & 6.836(0.16) & 1.237(0.06) & 1.535(0.07) & 2.133(0.09) & 6.836(0.16) \\ 
  PenSpline & 0.661(0.05) & 1.081(0.06) & 1.510(0.07) & 5.071(0.13) & {\bf1.072}(0.05) & {\bf1.109}(0.06) & {\bf1.514}(0.07) & {\bf5.232}(0.13) & 2.024(0.09) & {\bf1.231}(0.06) & {\bf1.607}(0.06) & 5.071(0.13) \\ 
  HS & {\bf0.573}(0.04) & 0.771(0.06) & 0.995(0.07) & 2.663(0.13) & 1.614(0.05) & 1.456(0.06) & 2.176(0.07) & 5.569(0.13) & {\bf0.921}(0.05) & 1.399(0.07) & 1.864(0.08) & {\bf4.514}(0.12) \\ 
  fHS1 & 0.578(0.04) & {\bf0.442}(0.04) & {\bf0.432}(0.04) & {\bf0.429}(0.04) & 1.627(0.05) & 1.551(0.07) & 2.114(0.09) & 6.463(0.15) & 1.230(0.06) & 1.499(0.07) & 2.055(0.09) & 5.915(0.14) \\ 
  fHS2 & 0.576(0.04) & {\bf0.442}(0.04) & {\bf0.432}(0.04) & {\bf0.429}(0.04) & 1.627(0.05) & 1.551(0.07) & 2.114(0.09) & 6.463(0.15) & 1.230(0.06) & 1.499(0.07) & 2.055(0.09) & 5.915(0.14) \\ 
  fHS3 & 0.576(0.04) & {\bf0.442}(0.04) & {\bf0.432}(0.04) & {\bf0.429}(0.04) & 1.627(0.05) & 1.551(0.07) & 2.114(0.09) & 6.463(0.15) & 1.230(0.06) & 1.499(0.07) & 2.055(0.09) & 5.915(0.14) \\ 
   \hline
\end{tabular}
}
\caption{ The results for the simple regression models. The smallest MSE is in bold for each $k_n$, except for the partial oracle estimator (``Oracle"). ``fHS1", ``fHS2", and ``fHS3" are the procedures based on the fHS prior with $b=\exp(-k_n\log n /10),\exp(-k_n\log n /4),$ and $\exp(-k_n\log n /2)$, respectively. }
\label{tab:simple}
\end{sidewaystable*}
  

For each prior, we used the posterior mean $\hat{f}$ as a point estimate for $f$, and reported the empirical {\it Mean Square Error} (MSE), i.e. $\norm{\widehat f - f}_{n,2}^2$. We also compare our approach to a partial oracle estimator enabled with the knnowledge of the functional form (parameteric or nonparametric) of the true function. When the true function has a parametric form, the partial oracle estimator is equivalent to the parametric estimator; otherwise, the partial oracle estimator is equivalent to the standard B-spline estimator.  

  
 Tables \ref{tab:simple} lists the MSE of the posterior mean estimator over 100 replicates in estimating the unknown function $f$ for the simple regression model with sample sizes of $n = 200$ and $500$. When the true function $f$ belongs to the nominal parametric class, the posterior mean function resulting from the fHS prior outperforms the HS prior. When the true  function does not belong to the class of the parametric functions,  the fHS prior performs comparably to the partial oracle estimator. 
 
 The penalized spline method  and the procedure based on the standard HS prior show smaller estimation error than that of the fHS prior and the partial oracle estimator (the standard B-spline estimator). This is because the penalized spline estimator  regularizes the smoothness of the function. In contrast,  the fHS prior produces a fitted function that is almost identical to the standard B-spline estimator in the nonlinear case. The shrinkage effect of the fHS prior towards a parametric function is only activated when the shape of the function fits the pre-specified parametric form. Thus, when the parametric model is true the fHS estimator behaves like the parametric estimator. If not, it behaves like the B-spline estimator.  

\section{Simulation studies for additive models} \label{sim:additive} 
Our regression examples in the previous subsection involved one predictor variable. In the case of multiple predictors, a popular modeling framework is the class of additive models  \citep{hastie1986generalized}, where the unknown function relating $p$ candidate predictors to a univariate response is modeled as the sum of $p$ univariate functions, with the $j$th function only dependent on the $j$-th predictor. In this section, we apply the fHS prior to additive models and compare results obtained under this prior to several alternative methods. To be consistent with our previous notation, we express additive models as  
\begin{eqnarray}\label{eq:additive}
Y = \sum_{j=1}^p F_j + \epsilon,
\end{eqnarray}
where $F_j = \{f_j(x_{1j}),\dots,f_j(x_{nj})  \}$ for $j=1,\dots,p$, and $\epsilon\sim\Gauss(0,\sigma^2\mr{I}_n)$. We let $\Phi_j$ denote the spline basis matrix for $X_j$ and let $\beta_j=\{ \beta_{j1},\dots,\beta_{jk_n} \}$ denote the corresponding coefficient. In general, each component function can be modeled nonparametrically. For example, using the B-spline basis functions as described in the previous section, $f_j(x) = \sum_{l=1}^{k_n}\beta_{jl}\phi_{l}(x)$,  so that $F_j=\Phi_j\beta_j$ for $j=1,\dots,p$. However, if there are many candidate predictors, then nonparametrically estimating $p$ functions may be statistically difficult and may result in a loss of precision and overfitting if only a small subset of the variables are significant. With this motivation, we extend the fHS framework to additive models, where we assign independent fHS priors to the $f_j$'s to facilitate shrinkage of each of these functions towards the class of   pre-specified parametric functions. For $\beta=\{\beta_1,\dots,\beta_p\}$, where $\beta_j\in\mathbb{R}^{k_n}$ for $j\in 1,\dots,p$, the resulting prior density  can be expressed as the product of independent fHS prior densities as follows: 
 \begin{eqnarray}\label{eq:additive_fHS}
\pi(\beta\mid \tau^2, \sigma^2 )&\propto& \prod_{j=1}^p\tau_j^{k_n-d_0}\exp\left\{  -\frac{\beta_j^\T\Phi_j^\T(\mr{I} - Q_{0j})\Phi_j\beta_j}{2\sigma^2\tau_{j}^2}\right\}\\
 \pi(\tau) &\propto& \prod_{j=1}^p\frac{(\tau_j^2)^{b-1/2 }}{(1+\tau_j^2)^{(a+b)}} \ind_{(0, \infty)}(\tau_j).
 \end{eqnarray}
 Here, $\tau=\{\tau_1,\dots,\tau_p\}$. This prior imposes  shrinkage on each $\beta_j^\T\Phi_j^\T(\mr{I} - Q_{0j})\Phi_j\beta_j$ towards zero so that the resulting posterior distribution contracts towards the class of the parametric functions. In particular, when $Q_{0j}=0$ for $j=1,\dots,p$, the resulting  posterior distribution on $F_j$ concentrates on the null function when the marginal effect of $F_j$ is negligible. This property enables us to  select variables by using the thresholding procedure discussed in Section \ref{sec:modelsel}. 
 

 For shrinkage across many variables, the classical HS prior includes a global shrinkage parameter common to all variables. In the present context, the role of the global shrinkage parameter is implicitly replaced by the scale parameter $b$. We treat $b$ as a fixed hyperparameter in the sequel and follow the default recommendation from the earlier section regarding its choice.
  
  For the univariate examples considered in the previous section, standard Bayesian model selection procedures based on the mixture of point mass priors \citep{choi2009note, choi2015note, choi2015partially}  can also be applied. These approaches have advantages in interpreting the results of model selection and Bayesian model averaging \citep{raftery1997bayesian}. However, when multiple functions are considered in model selection, standard procedures with discrete mixture priors can be computationally demanding in searching the discrete space of models.

\subsection{A comparison to the standard horseshoe prior}\label{sec:HS}
 Under the additive model,  one can impose a product of standard HS  priors \citep{carvalho2010horseshoe} on the spline coefficients to impose shrinkage towards the null function. The hierarchical structure of such an HS prior can be expressed as  
 \begin{eqnarray}\label{eq:HS_fHS}\nonumber
 \pi(\beta\mid \lambda, \psi, \sigma^2 )\propto \exp\left\{  -\frac{1}{\sigma^2\lambda^2} \sum_{j=1}^p\sum_{l=1}^{k_n}  \frac{\beta_{jl}^2}{\psi_{jl}^2}\right\}, \  \lambda \sim C^+(0,1), \ 
 \psi_{jl} \sim C^+(0,1),
 \end{eqnarray}
where $C^+(0,1)$ is the half-Cauchy distribution. 
The parameter $\lambda$ serves a global shrinkage parameter controlling the concentration near zero, while the  $\psi_{jl}$'s are local shrinkage parameters that control the tail heaviness of the individual coefficients \citep{polson2010}. The use of the standard HS prior  imposes strong shrinkage effects towards zero on each coefficient. But unlike the proposed fHS prior, the HS prior does not account for the grouping structure in the spline expansions of the components. We illustrate the importance of accounting for the group structure through a number of simulated and real examples next. We found that the fHS prior outperforms the vanilla HS prior. An analogy may be drawn to the superior performance of group lasso \citep{yuan2006model} over ordinary lasso when a similar group structure is present in the spline coefficients \citep{huang2010variable}. 

%
%

 It is not immediately clear how to select variables in an additive model by using the standard HS prior.  On the other hand, the thresholding procedure based on the fHS prior in \eqref{eq:additive_fHS} performs a natural model selection in this setting.

\subsection{Simulation scenarios}\label{sim:additive} 
For additive models, \cite{ravikumar2009sparse} proposed penalized likelihood procedures called {\it Sparse Additive Models} (SpAM) that combine ideas from model selection and additive nonparametric regression. The penalty term of SpAM can be described as  a weighted group Lasso penalty \citep{yuan2006model} in which  the coefficients for each component function $f_j$ for $j=1,\dots,p$ are forced to simultaneously shrink towards zero. \cite{meier2009high}  proposed the {\it High-dimensional Generalized Additive Model} (HGAM) that  differs from SpAM by its penalty term, which imposes both shrinkage towards zero and  regularization on the smoothness of the function. \cite{huang2010variable} introduced a two step procedure called  adaptive group Lasso (AdapGL) for additive models. The first step  estimates the weight of the group penalty, and the second  applies it to the adaptive group lasso penalty. Since the performance of penalized likelihood methods is sensitive to the choice of the tuning parameter, in the simulation studies that follow we considered  two criteria for tuning parameter selection: AIC and BIC. \texttt{R} packages \texttt{SAM}, \texttt{hgam}, and \texttt{grplasso} were used to implement SpAM, HGAM, and AdapGL, respectively. We also considered the standard HS prior and its computation was implemented by the \texttt{R} package \texttt{horseshoe}. We develop a blocked Gibbs sampler to fit the fHS procedure; the details are provided in Section E of the supplemental document. We observed good mixing and convergence of the algorithm developed based on examination of trace plots; see Section F of the supplemental document for representative examples. For the fHS prior and HS prior, we imposed a prior on $\sigma^2$ proportional to $1/\sigma^2$. We utilized 20,000 samples from the MCMC algorithms after 10,000 burn-in iterations to estimate the posterior mean.  
   
 We define the signal-to-noise ratio as $\mbox{SNR} =\Var(f(X))/Var(\epsilon)$, where $f$ is the true underlying regression function, i.e., $f=\sum_{j=1}^p f_j$, where $f_j$ is the true component function for $j=1,\dots,p$. We examine the same simulation scenarios that were considered in \cite{meier2009high} as follows:\\
\noindent {\bf Scenario 1:} ($p=200$, $\mbox{SNR}\approx 15$). This is the same scenario as Example 1 in \cite{meier2009high}.  A similar scenario was also considered in  
\cite{hardle2012nonparametric} and \cite{ravikumar2009sparse}. The true model is
$
Y_i = f_1(x_{i1}) + f_2(x_{i2}) + f_3(x_{i3}) + f_4(x_{i4}) + \epsilon_i,
$
where $\epsilon_i\overset{i.i.d}\sim\Gauss(0,1)$ for $i=1,\dots,n$, with $f_1(x) = -\sin(2x)$, $f_2(x) = x^2-25/12$, $f_3(x) = x$, $f_4(x) = \exp\{-x\} -2/5\cdot\sinh(5/2)$.
The covariates are independently generated from a uniform distribution between $-2.5$ to $2.5$. \\
\noindent {\bf Scenario 2:} ($p=80$, $\mbox{SNR}\approx 7.9$). This is equivalent to Example 3 in \cite{meier2009high} and similar to an example in \cite{lin2006component}. The true model is
$
Y_i = 5f_1(x_{i1}) + 3f_2(x_{i2}) + 4f_3(x_{i3}) + 6f_4(x_{i4}) + \epsilon_i,
$
where $\epsilon_i\overset{i.i.d}\sim\Gauss(0,1.74)$ for $i=1,\dots,n$, with
$f_1(x) = x$, $f_2(x) = (2x-1)^2$, $f_3(x) = \frac{\sin(2\pi x)}{2-\sin(2\pi x)}$, $f_4(x) = 0.1\sin(2\pi x) + 0.2\cos(2\pi x) + 0.3 \sin^2(2\pi x)+0.4 \cos^3(2\pi x) + 0.5 \sin^3(2\pi x)$.
 The covariate $\bx_j=(x_{1j},\dots,x_{nj})^\T$ for $j=1,\dots,p$ is generated by $
 \bx_j = (W_j + U)/2$, where $W_1,\dots,W_p$ and $U$ are independently simulated from U$(0, 1)$ distributions. \\
\noindent {\bf Scenario 3} ($p=60$, $\mbox{SNR}\approx 11.25$). This scenario is equivalent to Example 4 in 
  \cite{meier2009high}, and a similar example was also considered in \cite{lin2006component}. The same functions  and the same process to generate the covariates used in {\it Scenario 2} were used in this scenario. The true model is
$Y_i = f_1(x_{i1}) + f_2(x_{i2}) + f_3(x_{i3}) + f_4(x_{i4})  +1.5 f_1(x_{i5}) + 1.5 f_2(x_{i6}) + 1.5 f_3(x_{i7})  
 + 1.5 f_4(x_{i8}) + 2.5 f_1(x_{i9}) + 2.5 f_2(x_{i10}) + 2.5 f_3(x_{i11}) + 2.5 f_4(x_{i12}) +  \epsilon_i
$, where $\epsilon_i\overset{i.i.d}\sim\Gauss(0,0.5184)$ for $i=1,\dots,n$.  
  
  To evaluate the estimation performance of the fHS prior, we report the MSE for each method. To measure the performance of variable selection, we examined the proportion of times the true model was selected, as well as the {\it Matthews correlation coefficient} (MCC; \cite{matthews1975comparison}), defined as, 
  \be
  \mbox{MCC} = \frac{\mbox{TP}\cdot \mbox{TN} - \mbox{FP}\cdot \mbox{FN}}{(\mbox{TP}+\mbox{FP})(\mbox{TP}+\mbox{FN})(\mbox{TN}+\mbox{FP})(\mbox{TN}+\mbox{FN})},
  \ee
 where TP, TN, FP, and FN denote the number of true positive, true negatives, false positives, and false negatives, respectively. MCC is generally regarded as a balanced measure of the performance of classification methods, which simultaneously takes into account TP, TN, FP, and FN. We note that MCC is bounded by 1, and the closer MCC is to 1, the better the model selection performance is.  
  
\begin{sidewaystable}
\centering
\resizebox{19cm}{!}{%
\begin{tabular}{|c|ccc|ccc|ccc|ccc|}
   \hline
  &  \multicolumn{3}{|c|}{$k_n = 5$} & \multicolumn{3}{|c|}{$k_n = 8$} & \multicolumn{3}{|c|}{$k_n = 11$} &\multicolumn{3}{|c|}{$k_n = 35$} \\
  \hline
 & MSE & MCC & PT & MSE & MCC & PT & MSE & MCC & PT & MSE & MCC & PT \\ 
  \hline
Oracle ($n=300$) & 0.071(0.002) &  &  & 0.108(0.003) &  &  & 0.148(0.004) &  &  &0.452(0.006) & & \\ 
  HS & 0.185(0.005) &  &  & 0.562(0.076) &  &  & 2.646(0.114) &  &  & 0.809(0.007)  &  &\\ 
  fHS1 & 0.241(0.005) & 0.971(0.006) & 0.80 & 0.245(0.005) & 0.979(0.005) & 0.85 & {\bf0.298}(0.005) & 0.986(0.004) & 0.89 & 0.758(0.048) & 0.882(0.018) & 0.47\\ 
  fHS2 & 0.173(0.004) & 0.979(0.006) & {\bf0.86} & {\bf0.234}(0.004) & {\bf0.983}(0.006) & {\bf0.90} & 0.301(0.006) & 0.976(0.006) & 0.83 & {\bf 0.653}(0.012) & {\bf 0.937}(0.008) & 0.50\\ 
  fHS3 & 0.{\bf171}(0.004) & {\bf0.982}(0.005) & {\bf0.86} & 0.243(0.008) & {\bf0.983}(0.006) & 0.89 & {\bf0.298}(0.005) & {0.990}(0.003) & {0.92} & 0.708(0.024) & 0.918(0.011) & 0.50\\ 
  SpAM (AIC) & 0.992(0.089) & 0.897(0.097) & 0.37 & 0.360(0.007) & 0.679(0.010) & 0.00 & 0.394(0.007) & 0.514(0.008) & 0.00 & 2.12(0.057) & 0.310(0.005) & 0.00\\ 
  SpAM (BIC) & 1.286(0.093) & 0.932(0.009) & 0.54 & 1.899(0.072) & 0.984(0.004) & 0.87 & 2.051(0.060) & {\bf0.996}(0.002) & {\bf0.96} & 5.218(0.177) & 0.900(0.006) & 0.26\\ 
  HGAM (AIC) & 0.983(0.051) & 0.969(0.006) & 0.77 & 1.425(0.050) & 0.925(0.007) & 0.45 & 1.478(0.074) & 0.898(0.006) & 0.26 & 1.554(0.057) & 0.863(0.003) & 0.00\\ 
  HGAM (BIC) & 3.814(0.106) & 0.855(0.005) & 0.02 & 3.566(0.081) & 0.852(0.003) & 0.02 & 3.309(0.088) & 0.841(0.007) & 0.01 & 5.690(0.192) & 0.650(0.014) & 0.00\\ 
  AdapGL (AIC) & 0.197(0.005) & 0.343(0.006) & 0.00 & 0.277(0.006) & 0.280(0.003) & 0.00 & 0.352(0.006) & 0.258(0.003) & 0.00 & 0.706(0.007) & 0.326(0.001) & 0.00\\ 
  AdapGL (BIC) & 0.211(0.005) & 0.480(0.003) & 0.00 & 0.321(0.007) & 0.555(0.004) & 0.00 & 0.435(0.008) & 0.614(0.004) & 0.00 & 1.748(0.029) & 0.863(0.003) & 0.16\\ 
   \hline
Oracle ($n=600$)& 0.037(0.001) &  &  & 0.057(0.002) &  &  & 0.078(0.002) &  &  &  0.246(0.004) &  & \\ 
  HS & 0.072(0.002) &  &  & 0.133(0.002) &  &  & 0.191(0.004) &  &  & 0.682(0.005) &  &\\ 
  fHS1 & 0.092(0.002) & {0.984}(0.005) & {0.88} & 0.110(0.002) & {\bf0.986}(0.004) & 0.89 & 0.216(0.016) & 0.950(0.008) & 0.68 & {\bf 0.399}(0.004) & {\bf 0.996}(0.002) & 0.97\\ 
  fHS2 & 0.074(0.002) & {0.984}(0.004) & {0.88} & 0.108(0.003) & {\bf0.986}(0.005) & {\bf0.91} & 0.149(0.007) & 0.977(0.005) & 0.82 & 0.545(0.083) & 0.991(0.005) & 0.97\\ 
  fHS3 & {\bf0.073}(0.002) & {0.984}(0.004) & 0.87 & {\bf0.107}(0.002) & 0.985(0.005) & 0.89 & {\bf0.141}(0.003) & {0.983}(0.004) & {0.86} & 0.447(0.029) & 0.995(0.002) & 0.97 \\ 
  SpAM (AIC) & 1.080(0.092) & 0.927(0.008) & 0.50 & 0.228(0.029) & 0.720(0.010) & 0.02 & 0.207(0.004) & 0.532(0.008) & 0.00 & 1.011(0.025) & 0.302(0.009) & 0.00\\ 
  SpAM (BIC) & 1.105(0.093) & 0.929(0.008) & 0.51 & 1.145(0.094) & 0.928(0.008) & 0.51 & 1.783(0.065) & {\bf0.984}(0.004) & {\bf0.87} & 2.077(0.051) & 0.965(0.013) & 0.93\\ 
  HGAM (AIC) & 0.348(0.004) & {\bf1.000}(0.000) & {\bf1.00} & 0.762(0.033) & 0.868(0.003) & 0.04 & 0.989(0.047) & 0.861(0.002) & 0.00 & 1.376(0.089) & 0.826(0.007) & 0.00 \\ 
  HGAM (BIC) & 3.383(0.034) & 0.864(0.005) & 0.00 & 3.096(0.026) & 0.851(0.004) & 0.00 & 2.954(0.029) & 0.806(0.008) & 0.00 & 2.518(0.039) & 0.752(0.007) & 0.00\\ 
  AdapGL (AIC) & 0.129(0.093) & 0.693(0.008) & 0.00 & 0.152(0.003) & 0.457(0.007) & 0.00 & 0.183(0.003) & 0.342(0.005) & 0.00 & 0.428(0.004) & 0.234(0.001) & 0.00\\ 
  AdapGL (BIC) & 0.129(0.003) & 0.694(0.011) & 0.00 & 0.167(0.003) & 0.584(0.004) & 0.00 & 0.220(0.004) & 0.631(0.004) & 0.00 & 0.622(0.007) & 0.809(0.006) & 0.00\\ 
   \hline
\end{tabular}
}
\caption{Scenario 1. fHS1, fHS2, and fHS3 are the procedures based on the fHS prior with $b=\exp(-k_n\log n /10)$, $\exp(-k_n\log n /4)$, and $\exp(-k_n\log n /2)$, respectively. ``PT" is the proportion of times that each procedure selected the true model. The smallest MSE, and the largest MCC and PT are noted in bold for  each $k_n$. The smallest MSE and the largest MCC, except for the oracle estimator, are in bold. }   \label{tab:ex1}
\end{sidewaystable}


 We used the fHS prior in \eqref{eq:additive_fHS} with $Q_{0j}=0$ for all $j$. This setting of the fHS prior imposed a shrinkage effect towards the null function so that the posterior distribution of most component functions contracts towards zero. For model selection using the fHS prior, we selected variables with $E(\omega_j\mid Y)<1/2$ as described in Section \ref{sec:modelsel}, where $\omega_j = 1/(1+\tau_j^2)$ is the shrinkage coefficient for the $j$-th variable. To investigate the performance achieved by the proposed method, we compared it to a ``partial oracle estimator". The partial oracle estimator refers to the B-spline least squares estimator when the variables in the true model are given, but the true component functions in the additive model are not provided.

Results from simulation studies to compare these methods are depicted in Table \ref{tab:ex1} -- \ref{tab:ex3}. In most settings, the procedure based on the fHS prior has smaller MSE than the estimator based on the HS prior and the penalized likelihood estimators. These results hold consistently with different hyperparameters ($b =\exp(-k_n\log n /A)$ where $k_n\in\{5,8,11,35\}$ and $A\in\{2,4,10,35\}$). The SpAM with the tuning parameter chosen by BIC performs comparable to the  fHS procedure in some settings; for example, {\it Scenario 1}  with $k_n=11$, and {\it Scenario 2}  with $k_n=5$ and $k_n=11$. However, its estimation performance  is clearly inferior to the fHS procedure. The MSE of SpAM with BIC is at least two times larger than the estimator based on the fHS prior in all simulation scenarios. 

While the HS prior shows comparable estimation performance to the procedure based on the fHS prior in {\it Scenario 3},  its MSE is unstable and sensitive to the choice of $k_n$ in {\it Scenario 1} and {\it Scenario 2}. In particular, in {\it Scenario 1} with $k_n=11$,  
    the MSE of the HS prior is almost 9 times larger than the MSE of the fHS prior. When the number of basis function is chosen to be relatively large ($k_n=35$), the MSE of the fHS procedures with three different hyperparameters is uniformly smaller than that of the HS prior through all considered scenarios. In addition, as we have already discussed, model selection with the standard HS prior is not immediate in the present context.
    

\afterpage{%
    \clearpage
\begin{sidewaystable}
\centering
\resizebox{19cm}{!}{%
\begin{tabular}{|c|ccc|ccc|ccc|ccc|}
 \hline
  &  \multicolumn{3}{|c|}{$k_n = 5$} & \multicolumn{3}{|c|}{$k_n = 8$} & \multicolumn{3}{|c|}{$k_n = 11$} & \multicolumn{3}{|c|}{$k_n = 35$}\\
  \hline
 & MSE & MCC & PT & MSE & MCC & PT & MSE & MCC & PT &  MSE & MCC & PT\\ 
  \hline
Oracle ($n=300$) & 0.647(0.010) &  &  & 0.206(0.005) &  &  &  &  & & 0.812(0.006) &  &  \\ 
  HS & 0.730(0.011) &  &  & 0.473(0.010) &  &  & 0.604(0.013) &  & & 1.168(0.012) &  &\\ 
  fHS1 & {\bf0.654}(0.010) & 0.956(0.0070 & 0.68 & 0.293(0.006) & {\bf0.987}(0.004) & {\bf0.90} & 0.362(0.007) & 0.968(0.008) & 0.81 & {\bf 0.859}(0.025) & {\bf 0.861}(0.012) & 0.30\\ 
  fHS2 & 0.679(0.011) & {0.961}(0.006) & {0.72} & {\bf0.291}(0.006) & {\bf0.987}(0.004) & {\bf0.90} & {\bf0.359}(0.007) & {\bf0.980}(0.005) & {\bf0.86} & 0.938(0.048) & 0.839(0.014)  & 0.33\\ 
  fHS3 & 0.677(0.010) & 0.960(0.006) & 0.71 & 0.292(0.006) & 0.982(0.005) & 0.88 & 0.364(0.008) & 0.972(0.007) & 0.84 & 0.891(0.048) & 0.839(0.014) & 0.22\\ 
    SpAM (AIC) & 0.788(0.029) & 0.783(0.015) & 0.21 & 0.733(0.029) & 0.615(0.022) & 0.15 & 0.666(0.019) & 0.351(0.011) & 0.01 & 2.021(0.057) & 0.257(0.002) & 0.00\\ 
  SpAM (BIC) & 1.315(0.020) & {\bf0.985}(0.005) & {\bf0.89} & 1.479(0.092) & 0.964(0.007) & 0.76 & 2.535(0.207) & 0.865(0.018) & 0.53 & 5.901(0.086) & 0.700(0.002) & 0.00\\ 
  HGAM (AIC) & 0.786(0.009) & 0.922(0.007) & 0.46 & 0.457(0.008) & 0.868(0.005) & 0.10 & 0.504(0.017) & 0.810(0.008) & 0.02 & 1.197(0.052) & 0.696(0.002) & 0.00\\ 
  HGAM (BIC) & 1.802(0.037) & 0.712(0.006) & 0.03 & 1.404(0.041) & 0.717(0.006) & 0.02 & 1.601(0.078) & 0.685(0.006) & 0.00 & 6.869(0.103) & 0.490(0.003) & 0.00\\ 
  AdapGL (AIC) & 0.718(0.009) & 0.287(0.003) & 0.00 & 0.459(0.010) & 0.262(0.003) & 0.00 & 0.586(0.011) & 0.238(0.002) & 0.00 & 1.082(0.014) & 0.313(0.006) & 0.00\\ 
  AdapGL (BIC) & 0.967(0.016) & 0.571(0.004) & 0.00 & 0.636(0.016) & 0.624(0.005) & 0.00 & 0.882(0.021) & 0.694(0.007) & 0.00 & 2.105(0.026) & 0.711(0.013) & 0.00\\ 
   \hline
   
Oracle ($n=600$)& 0.622(0.007) &  &  & 0.101(0.002) &  &  & 0.127(0.003) &  &  &   0.407(0.005)&  &\\ 
  HS & 0.653(0.008) &  &  & 0.251(0.004) &  &  & 0.301(0.006) &  & & 0.881(0.007) &  &\\ 
  fHS1 & {\bf0.628}(0.008) & 0.976(0.005) & 0.83 & {\bf0.147}(0.002) & {\bf0.999}(0.001) & {\bf0.99} & {\bf0.182}(0.003) & {\bf0.997}(0.002) & {0.97} & {\bf 0.441}(0.006) & {\bf0.997}(0.002) & 0.98\\ 
  fHS2 & 0.638(0.007) & {0.982}(0.005) & {0.87} & {\bf0.147}(0.002) & {\bf0.999}(0.001) & {\bf0.99} & {0.183}(0.003) & 0.995(0.002) & 0.96 & 0.455(0.015) & 0.993(0.003)  & 0.95\\ 
  fHS3 & 0.638(0.008) & 0.977(0.005) & 0.83 & {\bf0.147}(0.002) & {\bf0.999}(0.001) & {\bf0.99} & 0.184(0.004) & 0.994(0.003) & 0.95 & 0.444(0.006) & 0.990(0.004) & 0.93\\ 
    SpAM (AIC) & 0.566(0.038) & 0.787(0.015) & 0.22 & 0.515(0.031) & 0.737(0.017) & 0.14 & 0.638(0.033) & 0.739(0.022) & 0.19 & 1.207(0.059) & 0.313(0.003) & 0.00\\ 
  SpAM (BIC) & 1.264(0.014) & {\bf1.000}(0.000) & {\bf1.00} & 1.214(0.017) & {\bf0.999}(0.001) & {\bf0.99} & 1.221(0.015) & {\bf0.997}(0.002) & {\bf0.98} & 5.564(0.148) & 0.694(0.008) & 0.65\\ 
  HGAM (AIC) & 0.953(0.032) & 0.748(0.008) & 0.01 & 0.469(0.005) & 0.854(0.004) & 0.01 & 0.327(0.004) & 0.810(0.008) & 0.00 & 0.598(0.025) & 0.698(0.002) & 0.00\\ 
  HGAM (BIC) & 1.779(0.017) & 0.698(0.004) & 0.00 & 1.474(0.018) & 0.700(0.001) & 0.00 & 1.316(0.019) & 0.701(0.002) & 0.00 & 2.971(0.048) & 0.501(0.005) & 0.00\\ 
  AdapGL (AIC) & 0.640(0.009) & 0.197(0.003) & 0.00 & 0.259(0.004) & 0.393(0.006) & 0.00 & 0.308(0.005) & 0.287(0.004) & 0.00 & 0.711(0.007) & 0.248(0.002) & 0.00\\ 
  AdapGL (BIC) & 0.800(0.009) & 0.568(0.006) & 0.00 & 0.323(0.005) & 0.615(0.004) & 0.00 & 0.427(0.007) & 0.639(0.004) & 0.00 & 1.412(0.016) & 0.877(0.009) & 0.28\\ 
   \hline
\end{tabular}
}
\caption{ Scenario 2. The description of this table is the same as Table \ref{tab:ex1}.}\label{tab:ex2}
\end{sidewaystable}
    \clearpage
}

\afterpage{%
    \clearpage
\begin{sidewaystable}
\centering
\resizebox{19cm}{!}{%
\begin{tabular}{|c|ccc|ccc|ccc|ccc|}
 \hline
  &  \multicolumn{3}{|c|}{$k_n = 5$} & \multicolumn{3}{|c|}{$k_n = 8$} & \multicolumn{3}{|c|}{$k_n = 11$} & \multicolumn{3}{|c|}{$k_n = 35$} \\
   \hline
 & MSE & MCC & PT & MSE & MCC & PT & MSE & MCC & PT & MSE & MCC & PT \\ 
  \hline
Oracle ($n=300$)& 0.214(0.003) &  &  & 0.169(0.003) &  &  & 0.230(0.003) &  &  & 0.520(0.004) &  &\\ 
  HS & {\bf0.253}(0.003) &  &  & 0.243(0.003) &  &  & {0.287}(0.003) &  &  &  0.477(0.005) &  &\\ 
  fHS1 & 0.263(0.003) & {\bf0.809}(0.006) & 0.00 & {\bf 0.237}(0.005) & {\bf0.799}(0.007) & 0.00 & {\bf0.287}(0.005) & {\bf0.754}(0.008) & 0.00  & {\bf0.417}(0.005) & {\bf0.612}(0.006) & 0.00\\ 
  fHS2 & 0.288(0.004) & 0.803(0.007) & 0.00 & 0.248(0.005) & 0.789(0.007) & 0.00 & 0.292(0.007) & 0.746(0.007) & 0.00  & 0.428(0.007) & 0.609(0.009) & 0.00\\ 
 fHS3 & 0.286(0.004) & 0.796(0.006) & 0.00 & 0.241(0.004) & 0.796(0.007) & 0.00 & {0.287}(0.005) & 0.751(0.006) & 0.00  & 0.426(0.007) & {\bf0.612}(0.008) & 0.00\\ 
   SpAM (AIC) & 0.817(0.024) & 0.777(0.006) & 0.00 & 0.718(0.021) & 0.727(0.008) & 0.00 & 0.627(0.016) & 0.632(0.010) & 0.00  & 1.918(0.031) & 0.217(0.011) & 0.00\\ 
  SpAM (BIC) & 1.423(0.053) & 0.741(0.006) & 0.00 & 1.778(0.070) & 0.701(0.008) & 0.00 & 2.730(0.095) & 0.613(0.008) & 0.00  & 4.771(0.096) & 0.402(0.007) & 0.00\\ 
  HGAM (AIC) & 0.275(0.003) & 0.601(0.006) & 0.00 & 0.217(0.003) & 0.521(0.004) & 0.00 & 0.272(0.004) & 0.485(0.005) & 0.00  & 1.257(0.037) & 0.297(0.008) & 0.00\\ 
  HGAM (BIC) & 0.903(0.045) & 0.403(0.010) & 0.00 & 1.399(0.073) & 0.304(0.013) & 0.00 & 2.145(0.122) & 0.198(0.018) & 0.00  & 5.303(0.133) & 0.019(0.008) & 0.00\\ 
  AdapGL (AIC) & 0.289(0.003) & 0.445(0.007) & 0.00 & 0.247(0.003) & 0.376(0.006) & 0.00 & {\bf0.284}(0.003) & 0.323(0.007) & 0.00  & 0.513(0.007) & 0.359(0.008) & 0.00\\ 
  AdapGL (BIC) & 0.472(0.008) & 0.669(0.008) & 0.00 & 0.492(0.007) & 0.658(0.008) & 0.00 & 0.648(0.009) & 0.675(0.007) & 0.00  & 2.552(0.033) & 0.582(0.006)  & 0.00\\ 
   \hline
  Oracle ($n=600$)& 0.176(0.002) &  &  & 0.085(0.001) &  &  & 0.115(0.001) &  &   & 0.364(0.003) &  &\\ 
  HS & 0.203(0.002) &  &  & 0.149(0.002) &  &  & 0.187(0.002) &  &   & 0.424(0.004) &  &\\ 
  fHS1 & {\bf0.201}(0.002) & {\bf0.913}(0.005) & 0.07 & {\bf0.123}(0.002) & {\bf0.934}(0.004) & {\bf0.15} & 0.153(0.004) & 0.914(0.005) & 0.08  & 0.327(0.007) & 0.752(0.005) & 0.00\\ 
  fHS2 & 0.207(0.002) & 0.904(0.005) & 0.05 & 0.125(0.002) & 0.926(0.004) & 0.12 & 0.148(0.002) & 0.916(0.005) & 0.08 & {\bf0.306}(0.005) & {\bf0.757}(0.005) & 0.00\\ 
   fHS3 & 0.207(0.002) & 0.911(0.005) & {\bf0.08} & 0.125(0.002) & 0.931(0.004) & 0.13 & {\bf0.147}(0.002) & {\bf0.919}(0.005) & {\bf0.09} & 0.312(0.005) & 0.755(0.005) & 0.00\\ 
  SpAM (AIC) & 0.542(0.019) & 0.855(0.005) & 0.00 & 0.475(0.013) & 0.845(0.006) & 0.00 & 0.475(0.015) & 0.810(0.007) & 0.00  & 0.529(0.095) & 0.234(0.005) & 0.00\\ 
  SpAM (BIC) & 0.642(0.030) & 0.849(0.005) & 0.00 & 0.720(0.039) & 0.839(0.006) & 0.00 & 1.068(0.046) & 0.793(0.007) & 0.00  & 3.883(0.174) & 0.459(0.006) & 0.00\\ 
  HGAM (AIC) & 0.332(0.007) & 0.361(0.006) & 0.00 & 0.195(0.005) & 0.455(0.006) & 0.00 & 0.154(0.002) & 0.386(0.004) & 0.00 & 0.368(0.005) & 0.248(0.007) & 0.00\\ 
  HGAM (BIC) & 0.516(0.007) & 0.288(0.006) & 0.00 & 0.481(0.019) & 0.262(0.007) & 0.00 & 0.706(0.045) & 0.187(0.014) & 0.00  & 2.603(0.014) & 0.066(0.012) & 0.00\\ 
  AdapGL (AIC) & 0.205(0.002) & 0.275(0.005) & 0.00 & 0.401(0.005) & 0.010(0.006) & 0.00 & 0.197(0.002) & 0.405(0.007) & 0.00  & 0.364(0.003) & 0.310(0.006) & 0.00\\ 
  AdapGL (BIC) & 0.294(0.004) & 0.688(0.005) & 0.00 & 0.251(0.004) & 0.729(0.006) & 0.00 & 0.365(0.006) & 0.759(0.007) & 0.00 & 1.229(0.015) & 0.741(0.004) & 0.00\\ 
   \hline
\end{tabular}
}
\caption{ Scenario 3. The description of this table is the same as Table \ref{tab:ex1}.}\label{tab:ex3}
\end{sidewaystable}
    \clearpage
}


\section{Real data analysis for sparse additive model under high-dimensional settings}

 In this section, we considered the Near Infrared (NIR) Spectroscopy data set to examine the performance of the fHS prior for sparse additive models in high-dimensional settings. This data set was previously analyzed in \cite{liebmann2009determination} and \cite{curtis2014fast}, and is  available in the \texttt{R} package \texttt{chemometrics}. The NIR data includes glucose and ethanol concentration (in g/L) for 166 alcoholic fermentation mashes of different feedstock (rye, wheat and corn). Two hundred thirty-five NIR spectroscopy absorbance values were acquired in the wavelength range of 115-2285 nanometer (nm) by a transflectance probe \citep{liebmann2009determination}. We implemented the  model selection procedure on the data values with a response variable defined by ethanol concentrations. We have $n=166$ and $p=235$; we set the training and test set sizes to be 146 and  20, respectively. For each method, the prior specification used in Section \ref{sim:additive} was applied. Results are summarized in Table \ref{tab:NIR} and show that the proposed procedure with the fHS prior achieves the smallest prediction error among the considered methods. In addition, the average model size of the fHS procedure was smaller than that selected by the other methods. Compared to other procedures, the fHS procedure shows  stable performance overall. This result typically held regardless of the choice of $b$ and $k_n$. The exception occurred when  $b = \exp(-k_n\log n/10)$ and $k_n=5$. In that case, the average model size was $26.37$,  almost double that compared to the fHS procedure with the other hyperparameter values. One remark is that when $k_n=35$, all procedures showed poor and unstable prediction performances, except for  the HGAM procedures. We think that this is because the HGAM imposes extra regularization on the smoothness of the function, unlike other procedures. So, the corresponding HGAM estimator avoids an overfitting issue caused by a relatively large $k_n$.

\begin{table}[ht]
\centering
\resizebox{16cm}{!}{%
\begin{tabular}{|c|rc|rc|rc|rc|}
  \hline
  &  \multicolumn{2}{|c|}{$k_n = 5$} & \multicolumn{2}{|c|}{$k_n = 8$} & \multicolumn{2}{|c|}{$k_n = 11$} & \multicolumn{2}{|c|}{$k_n = 35$} \\
 \hline
  & MSPE & MS & MSPE & MS & MSPE & MS & MSPE & MS\\ 
  \hline
HS & 1.542(0.14) &  & 3.604(0.72) &  & 6.724(1.01) &  & 50.673(4.81) &\\ 
  fHS1 & 1.450(0.15) & 26.37 & {\bf2.014}(0.22) & 17.27 & 3.712(0.96) & 12.99 & 810.826(42.86) & 4.80\\ 
  fHS2 & 1.637(0.16) & 13.57 & 2.052(0.27) & 15.33 & {\bf2.521}(0.46) & 12.73 & 73.996(30.40) & 4.78\\ 
  fHS3 & {\bf1.446}(0.14) & 13.78 & 2.222(0.41) & 14.39 & 2.970(0.80) & 12.45 & 27.400(4.18) & 4.50\\ 
  SpAM (AIC) & 13.977(1.38) & 38.93 & 24.707(2.36) & 27.89 & 28.683(2.71) & 15.94 & 111.218(11.12) & 2.96\\ 
  SpAM (BIC) & 49.294(6.06) & 36.54 & 65.957(7.84) & 24.32 & 60.924(7.97) & 13.86 & 146.869(14.69) & 2.78\\ 
  HGAM (AIC) & 2.036(0.13) & 39.69 & 2.286(0.24) & 33.07 & 2.776(0.29) & 34.17 & {\bf 3.911}(0.39) & 21.60\\ 
  HGAM (BIC) & 1.854(0.12) & 45.19 & 2.285(0.24) & 32.86 & 2.786(0.30) & 33.91 & 3.912(0.39) & 21.50\\ 
  AdapGL (AIC) & 19.914(3.57) & 38.40 & 47.016(8.09) & 109.93 & 57.948(8.05) & 79.80 & 75.519(8.53) & 7.80\\ 
  AdapGL (BIC) & 10.626(1.42) & 14.07 & 16.370(2.57) & 15.06 & 33.421(4.45) & 14.25 & 476.551(12.96) & 0.00\\ 
   \hline
\end{tabular}
}
\caption{ NIR data set. ``MS" indicates the average model size. The smallest MSPE is noted in bold.}
\label{tab:NIR}
\end{table}


\section{Conclusion}
  We have proposed a class of shrinkage priors which we call the fHS priors. These priors impose strong shrinkage towards a pre-specified class of functions. The shrinkage mechanism in this prior is new. It allows the nonparametric function to shrink towards a parametric function without performing selection or shrinkage on the basis coefficients towards zero. By doing so, it preserves the minimax optimal parametric rate of posterior convergence $n^{-1/2}$ when the true underlying function is parametric. It also comes within $O(\log n)$ of achieving the minimax nonparametric rate  when the true function is strictly separated from the class of parametric functions. We also investigated the asymptotic properties of model selection procedure by thresholding  the posterior mean of $\omega$. The resulting model selection procedure consistently selects the true form of the regression function  as $n$ increases.

   The fHS prior imposes shrinkage on the shape of the function rather than shrinking or selecting certain basis coefficients. Hence, its scope of applicability is broad and it can be applied whenever a distance function to the null subspace can be formulated. In contrast, standard selection/shrinkage priors need an explicit parameterization of the null space in terms of zero constraints on specific parameters/coefficients.

    Like other nonparametric procedures, it is important to choose an appropriate value of the hyperparameters of the fHS prior ($k_n$ for the B-spline basis and $b$ for the hyperprior on $\omega$). In the real and simulated examples considered here, we used  multiple hyperparameters, $k_n \in \{5,8,11,35\}$ and $b = \exp(-k_n\log n /B)$ with $B \in \{10, 4, 2\}$, and compared the results with the different choice of the hyperparameters. More formal criterion to choose $k_n$ or $b$ might also be considered, and investigation of such criterion remains an active area of research.

  The novel shrinkage term contained in the proposed prior, $F^\T\IP F$, can be naturally applied to a new class of  penalized likelihood methods having a general form expressible as  $ - l(Y\mid F) + p_{\lambda}\big( F^\T\IP F \big)$, where $l(Y\mid F)$ is the logarithm of a likelihood function and $p_\lambda$ is the penalty function. In contrast to other penalized likelihood methods, this form of penalty allows shrinkage towards the space spanned by a projection matrix $Q_0$, rather than simply a zero function.

      \section*{Acknowledgment}
      All authors acknowledge support from NIH grant CA R01 158113. 
     \section*{Supplementary Material}
     The supplementary material, which is available online,  contains additional simulated and real data examples, MCMC diagnostics, and proofs of the theoretical results. In Section A in the supplementary material, a detailed description of the B-spline basis function is provided. In Section B, we examine additional simulation studies for univariate examples. These examples include the varying coefficient model and the log-density model introduced in Section \ref{sec:sim}. In Section C, we provide additional real data examples for the additive model. Section D contains the proofs of the  theoretical results. In Section E and F, the MCMC algorithm used to implement the fHS procedure is described and its convergence diagnostics is examined, respectively.    

  
\newpage
\bibliography{myReference.bib}

\end{document}